\documentclass[journal,onecolumn]{IEEEtran}
\usepackage{longtable}
\usepackage{multirow}
\usepackage{tabularx}
\usepackage{lipsum}
\usepackage{multicol}
\usepackage{amssymb}
\usepackage{amsmath}
\usepackage{longtable}
\newtheorem{theorem}{Theorem}[section]
\newtheorem{lemma}[theorem]{Lemma}
\newtheorem{proposition}[theorem]{Proposition}
\newtheorem{corollary}[theorem]{Corollary}
\newtheorem{example}{Example}
\usepackage{xcolor}

\newenvironment{proof}{\begin{IEEEproof}}{\end{IEEEproof}}

\newenvironment{definition}[1][Definition]{\begin{trivlist}
\item[\hskip \labelsep {\bfseries #1}]}{\end{trivlist}}

\newenvironment{remark}[1][Remark]{\begin{trivlist}
\item[\hskip \labelsep {\bfseries #1}]}{\end{trivlist}}

\makeatletter

\newcommand{\Rmnum}[1]{\expandafter\@slowromancap\romannumeral #1@}
\makeatother

\ifCLASSINFOpdf
\else
\fi
\begin{document}
%
\title{
On  $\sigma$-LCD  codes}

\author{Claude Carlet$^1$ \and Sihem Mesnager$^2$ \and Chunming Tang$^3$ \and Yanfeng Qi$^4$
\thanks{This work was supported by the project SECODE and by
the National Natural Science Foundation of China
(Grant No. 11401480, 11531002, 11371138). C. Tang
also acknowledges support from 14E013 and
CXTD2014-4 of China West Normal University.
Y. Qi also acknowledges support from Zhejiang provincial Natural Science Foundation of China (LQ17A010008, LQ16A010005).
}

\thanks{C. Carlet is with Department of Mathematics, Universities of Paris VIII and XIII, LAGA, UMR 7539, CNRS, Sorbonne Paris Cit\'{e}. e-mail: claude.carlet@univ-paris8.fr}

\thanks{S. Mesnager is with Department of Mathematics, Universities of Paris VIII and XIII and Telecom ParisTech, LAGA, UMR 7539, CNRS, Sorbonne Paris Cit\'{e}. e-mail: smesnager@univ-paris8.fr}

\thanks{C. Tang is with School of Mathematics and Information, China West Normal University, Nanchong, Sichuan,  637002, China. e-mail: tangchunmingmath@163.com
}

\thanks{Y. Qi is with School of Science, Hangzhou Dianzi University, Hangzhou, Zhejiang, 310018, China.
e-mail: qiyanfeng07@163.com
}

}

%


\maketitle

\begin{abstract}
 Linear complementary pairs (LCP) of codes play an important role in armoring implementations against side-channel attacks and fault injection attacks.
 One of the most common ways to construct LCP of codes is to use Euclidean linear complementary dual (LCD) codes. In this paper, we first introduce
 the concept of linear codes with
 $\sigma$ complementary dual ($\sigma$-LCD), which includes known Euclidean LCD codes, Hermitian LCD codes, and Galois LCD codes.
 As Euclidean LCD codes, $\sigma$-LCD codes can also be used to construct LCP of codes. We show that, for $q > 2$, all q-ary linear codes are $\sigma$-LCD and that, for every binary linear code $\mathcal C$, the code $\{0\}\times \mathcal C$ is  $\sigma$-LCD.
  Further, we study deeply $\sigma$-LCD generalized quasi-cyclic (GQC) codes. In particular,  we provide  characterizations of $\sigma$-LCD GQC codes, self-orthogonal GQC codes and  self-dual GQC codes, respectively. Moreover, we provide constructions of asymptotically good $\sigma$-LCD GQC codes. Finally, we focus on $\sigma$-LCD Abelian codes and prove  that all Abelian codes in a semi-simple group algebra are $\sigma$-LCD.
 The results derived in this paper extend those on the classical LCD codes and
 show that $\sigma$-LCD codes allow the construction of  LCP of codes more easily and with more flexibility.

\end{abstract}


%
\IEEEpeerreviewmaketitle

\section{Introduction}
Throughout this paper, let $\mathbb F_q$ be a finite field of size $q$, where $q$ is a prime power.
The Euclidean inner product $<\cdot, \cdot>$ over $\mathbb F_q^n$ is defined as
$<\mathbf b, \mathbf c>=\sum_{i=1}^n b_i c_i$,
where $\mathbf b=(b_1,\cdots, b_n), \mathbf c=(c_1,\cdots, c_n)\in \mathbb F_q^n$.
An $[n, k, d]$ linear code $\mathcal C$ over $\mathbb F_q$ is a linear subspace of $\mathbb F_q^n$ with dimension $k$ and minimum (Hamming) distance $d$.
 Its Euclidean dual code, denoted by $\mathcal C^{\perp}$, is defined by
$ \mathcal C^{\perp}=\{\mathbf b\in \mathbb F_q^n: <\mathbf b, \mathbf c>=0, \text{ for all }\mathbf c\in \mathcal C\}.$

The linear code $\mathcal C$ is called a  linear code with complementary dual (Euclidean LCD) if $\mathcal C +\mathcal C^{\perp}=\mathbb F_q^n$.
If $\mathcal C_1$ and $\mathcal C_2$ are linear codes over $\mathbb F_q$ with length $n$ and respective dimensions $k$ and $n-k$  such that $\mathcal C_1 +\mathcal C_2= \mathbb F_q^n$, the pair $(\mathcal C_1, \mathcal  C_2)$  will be called a linear complementary pair (LCP) of codes  with parameters $[n,k, d_1, d_2]$, where $d_1$ and $d_2$ are  the minimum distances of $\mathcal C_1$ and $\mathcal C_2^{\perp}$ respectively. Recently, Bringer et al. \cite{BCCGM14} introduced an interesting application of binary LCP of codes ($\mathcal C_1$, $\mathcal C_2$) to protect the implementations of symmetric cryptosystems against side channel attacks (SCA) and against fault injection attacks (FIA). The level of protection against SCA depends on the minimal distance of the second code while the level of protection against FIA depends on the minimal distance of the first code. Carlet and Guilley \cite{CG14} investigated further LCD codes from this perspective and extended it to the prevention from hardware Trojan horse (HTH) insertions (that is, gates added by for instance a silicon foundry into the design at fabrication time). Carlet et al. \cite{CDDGNNPT15} studied further this latter perspective in the general framework of LCP of codes.  In this perspective, a large enough minimum distance of the first code ensures the detection of HTH modifications and a large enough minimal distance of the second code ensures that HTH connected registers or nodes will not retrieve useful information.
 One of the most common ways to construct LCP of codes is to use Euclidean LCD codes.
Specifically,
if $\mathcal C$ is an Euclidean $[n, k, d]$ LCD code, $(\mathcal  C, \mathcal C^{\perp})$
is an LCP of codes with parameters $[n,  k, d, d]$.

A lot of works has been devoted to the characterization and constructions of Euclidean LCD codes. Yang and Massey provided in \cite{YM94} a necessary and sufficient condition under which a cyclic code has a complementary dual. In \cite{DLL16}, Ding et al. constructed several families of Euclidean LCD cyclic codes over finite fields and analyzed their parameters. In \cite{LDL16} Li et al. studied a class of Euclidean LCD BCH codes proposed in \cite{LDL16V0} and extended the results on their parameters. In \cite{GZS16}, quasi-cyclic (QC) codes that are Euclidean LCD have been characterized and studied using their concatenated structures. Criteria for Euclidean complementary duality of generalized quasi-cyclic codes (GQC) bearing on the component codes were  given and some explicit long GQC that are Euclidean LCD, but not quasi-cyclic, have been exhibited in \cite{GOOSSS17}. In \cite{CW16}, Cruz and Willems
investigated and characterized ideals in a group algebra which have Euclidean
complementary duals.
Mesnager et al. \cite{MTQ16} provided a construction of algebraic geometry Euclidean LCD codes which could be good candidates to be resistant against SCA.
In \cite{Jin16} and \cite{CMTQ17mds}, the construction of Euclidean and Hermitian LCD codes  was studied.

However, little is known on Hermitian LCD codes. More precisely, it has been proved in \cite{GZS16} that those codes are asymptotically good. By employing their generator matrices, Boonniyoma and Jitman gave in \cite{BJ16} a sufficient and necessary condition for Hermitian codes to be LCD. Li \cite{Li17} constructed some cyclic Hermitian LCD codes over finite fields and analyzed their parameters. In \cite{CMTQ17}, the authors completely determined  all $q$-ary ($q>3$) Euclidean LCD codes and all $q^2$-ary ($q>2$) Hermitian LCD codes for all possible parameters.
In \cite{LFL17}, Liu et al. studied complementary Galois dual codes.

In this paper, we first introduce
 the concept of
 $\sigma$-LCD codes, which includes known Euclidean LCD codes, Hermitian LCD codes, and Galois LCD codes.
 As Euclidean LCD codes, $\sigma$-LCD codes can also be used to construct LCP of codes. We show that, all $q$-ary ($q>2$) linear codes are $\sigma$
 LCD and $\{0\}\times \mathcal C$ is   $\sigma$-LCD, where $\mathcal C$ is any binary linear code. Further, we  study $\sigma$-LCD GQC  codes. Finally, we consider $\sigma$-LCD Abelian codes and prove  that
 all Abelian codes in semi-simple group algebra are $\sigma$-LCD.
 The results of this paper show that $\sigma$-LCD codes allow us to construct LCP of codes more easily and with more flexibility.

 The paper is organized as follows. In Section \ref{sec:sigma-LCD-codes},  we introduce the concept of $\sigma$-LCD codes and develop some general results on LCP of codes.
 In Section \ref{sec:LCD-GQC}, we firstly characterize $\sigma$-LCD GQC codes. Based on these results we investigate $\sigma$-LCD $1$-generator  GQC codes. In addition, we also show
 that some classes of  $\sigma$-LCD codes are asymptotically good.
 In Section \ref{sec:Abelian-LCD}, we investigate and characterize ideals (or Abelian codes) in a communicative group  algebra which have $\sigma$ complementary duals.

\section{$\sigma$-LCD codes and LCP of codes}\label{sec:sigma-LCD-codes}

In this section, we present the definition of $\sigma$ inner product,
which generalizes the one of Euclidean product, Hermitian product, and of Galois product \cite{FZ16}. From the $\sigma$ inner product, we introduce the  concept of
$\sigma$-LCD codes and provide some general results on $\sigma$-LCD codes
and LCP of codes.

In the following, $\mathbf{Aut}(\mathbb F_q^n)$ denotes the group of all mappings $\sigma$
from $\mathbb F_q^n$ to  $\mathbb F_q^n$, where $\sigma$ satisfies the following three conditions:

(i) $\sigma(\mathbf c+\mathbf w)= \sigma(\mathbf c)+\sigma(\mathbf w)$ for any $\mathbf c, \mathbf w \in \mathbb F_q^n$, that is, $\sigma$ is an  $\mathbb F_p$-linear transformation;

(ii) $d_{\mathbf{H}}(\sigma(\mathbf c), \sigma(\mathbf w))=d_{\mathbf{H}}(\mathbf c,\mathbf w)$ for any $\mathbf c, \mathbf w \in \mathbb F_q^n$, that is, $\sigma$ preserves the Hamming distances.

(iii) for any $\mathbb F_q$ linear code $\mathcal C$, $\sigma(\mathcal  C)$ is also an $\mathbb F_q$ linear code.

The set of coordinate permutations over $\mathbb F_q^n$  forms a group, which is referred to as the \emph{permutation  group} of $\mathbb F_q^n$ and denoted by $\mathbf{PAut}(\mathbb F_q^n)$.
The set of diagonal transformations over $\mathbb F_q^n$  (that is, transformations for which the associated matrices are diagonal) forms the group $\mathbf{DAut}(\mathbb F_q^n)$, which is called the \emph{diagonal group} of $\mathbb F_q^n$.
A monomial matrix is a square matrix with exactly one nonzero entry in each row and column.
A monomial matrix $\sigma$ can be written either in the form $\gamma_1 \pi$ or the form $\pi \gamma_2$ , where $\gamma_1$ and $\gamma_2$ are diagonal matrices and $\pi$
 is a permutation matrix. The set of
 transforms with  monomial matrices  forms the group $\mathbf{MAut}(\mathbb F_q^n)$, which is called the \emph{monomial  group}.
 The sets $\mathbf{PAut}(\mathbb F_q^n)$, $\mathbf{DAut}(\mathbb F_q^n)$ and $\mathbf{MAut}(\mathbb F_q^n)$  are subgroups of $\mathbf{Aut}(\mathbb F_q^n)$. The reader notices that all the mappings $\sigma$ in $\mathbf{PAut}(\mathbb F_q^n)$, $\mathbf{DAut}(\mathbb F_q^n)$ and $\mathbf{MAut}(\mathbb F_q^n)$ satisfy the conditions (i), (ii) and (iii) (every mapping in them can be written  as $AB$, where $A$ is a permutation and $B$ is a diagonal mapping and it is clear that $A$ and $B$ satisfy the conditions (i), (ii) and (iii)).
 In the binary case,
$\mathbf{DAut}(\mathbb F_q^n)$ contains only identity transformation and $\mathbf{PAut}(\mathbb F_q^n)= \mathbf{MAut}(\mathbb F_q^n)$.

Let $\sigma\in \mathbf{Aut}(\mathbb F_q^n)$ and $\mathbf w, \mathbf c \in \mathbb F_q^n$. We define the $\sigma$ \emph{inner product}, as follows:
\begin{align*}
<\mathbf w, \mathbf c>_{\sigma}=<\mathbf w, \sigma(\mathbf c)>=\sum_{i=0}^{n-1} w_i d_i,
\end{align*}
where $\mathbf w=(w_0,w_1, \cdots, w_{n-1})$ and $\sigma(\mathbf c)=(d_0,d_1, \cdots, d_{n-1})$.

If $\sigma$ is the identity transformation, then the $\sigma$ inner product coincides with the usual Euclidean inner product. Moreover,  if $\sigma$ is the mapping defined by
$\sigma(c_0,c_1,\cdots, c_{n-1})=(c_0^{\sqrt{q}},c_1^{\sqrt{q}},\cdots, c_{n-1}^{\sqrt{q}})$ where $q$ is a square, then the $\sigma$ inner product coincides the Hermitian inner product.
However if $\sigma$ is a transformation of $\mathbb F_q^n$ induced  by  Frobenius mapping over $\mathbb F_q$, then the $\sigma$ inner product corresponds to Galois inner product \cite{FZ16}.

Let $\sigma\in \mathbf{Aut}(\mathbb F_q^n)$ and $\mathcal C$ be a linear  code over $\mathbb F_q$ with length $n$. We call
\begin{align*}
\mathcal C^{\perp_{\sigma}}=\{\mathbf w\in \mathbb F_q^n: <\mathbf w, \mathbf c>_{\sigma}=0 \text{ for any } \mathbf c\in \mathcal C \}
\end{align*}
as the $\sigma$ \emph{dual} code of $\mathcal C$. From the definition of $\mathcal C^{\perp_{\sigma}}$,  one gets immediately the following relationship
\begin{align}\label{eq:sigma-euclidean-dual}
\mathcal C^{\perp_{\sigma}}=(\sigma(\mathcal C))^{\perp},
\end{align}
where $\sigma(\mathcal C)=\{\sigma(\mathbf c):  \mathbf c \in \mathcal C\}$ and
$(\sigma(\mathcal C))^{\perp}$ is the dual code of $\sigma(\mathcal C)$.

It is easy to see that $\sigma$ inner product is nondegenerate,  $\mathcal C^{\perp_{\sigma}}$ is also an $\mathbb F_q$ linear subspace of $\mathbb F_q^n$, and
$dim_{\mathbb F_q}(\mathcal C)+ dim_{\mathbb F_q}(\mathcal C^{\perp_{\sigma}}) = n$.

\begin{definition}
 A linear code $\mathcal C$ over $\mathbb F_q^n$ is said to be $\sigma$ complementary dual
 (\emph{$\sigma$-LCD})
 if $\mathcal C \cap \mathcal C^{\perp_{\sigma}}=\{0\}$.
 \end{definition}
 If $\mathcal C$ is an $[n, k, d]$ $\sigma$-LCD code, $(\mathcal  C, \mathcal C^{\perp_{\sigma}})$
is an LCP of codes with parameters $[n,  k, d, d]$.
A linear code $\mathcal C$ over $\mathbb F_q^n$ shall be called
 a $\sigma$ \emph{self-orthogonal} code (resp. $\sigma$ \emph{self-dual} code) if  $\mathcal C \subseteq  \mathcal C^{\perp_{\sigma}}$ (resp. if $\mathcal C =  \mathcal C^{\perp_{\sigma}}$). It is easy to see that any $\sigma$  self-dual code has dimension $\frac{n}{2}$.

Given a matrix G, $Rank(G)$ (resp. $G^{T}$) denotes the rank of $G$ (resp. the transpose of $G$).
We start with a general result on the dimension of $\mathcal C \cap \mathcal C^{\perp_\sigma}$.
\begin{proposition}\label{prop:dim-hull}
Let $\sigma\in \mathbf{Aut}(\mathbb F_q^n)$ and $\mathcal  C$ be a $q$-ary $[n,k,d]$ code with  generator matrix $G$. Then,
$dim_{\mathbb F_q}(\mathcal C \cap \mathcal C^{\perp_\sigma})=k-Rank\left (G(G^{\sigma})^{T} \right )$.
\end{proposition}
\begin{proof}
We denote by $G^{\sigma}$  the $k\times n$ matrix with its $i$-th row being $\sigma(G(i,:))$ for $i\in \{1,2, \cdots, k\}$, where $G(i,:)$ is the $i$-th row of $G$.
Let $\sum_{i=1}^k x_i G(i,:)\in \mathcal C $, where $x_i\in \mathbb F_q$. Then, $\sum_{i=1}^k x_i G(i,:)\in \mathcal C  \cap \mathcal C^{\perp_\sigma}$ if
and only if
\begin{align*}
<\sum_{i=1}^k x_i G(i,:), \sigma(G(j,:))>=0 \text{ for any } j\in \{1, \cdots, k\},
\end{align*}
 that is
 \begin{align*}
 G(G^{\sigma} )^{T} \mathbf v=0,
 \end{align*}
where $\mathbf v=(x_1, \cdots, x_k)^{T}$. Hence, $dim_{\mathbb F_q}(\mathcal C \cap \mathcal C^{\perp_\sigma})=k-Rank(G(G^{\sigma})^{T})$, which completes the proof.
\end{proof}

Given a linear code  $\mathcal C$, there exists  $\sigma\in \mathbf{PAut}(\mathbb F_q^n)$ for which one can identify the generator matrix of the code $\sigma(\mathcal C)$ as well as $\sigma(\mathcal C) \cap (\sigma(\mathcal C))^{\perp}$.

\begin{lemma}\label{lem:hull}
Let $\mathcal C$ be a $q$-ary $[n,k,d]$ linear code with $h=dim_{\mathbb F_q}(\mathcal C \cap \mathcal C^{\perp})>0$. Then, there exists
$\sigma\in \mathbf{PAut}(\mathbb F_q^n)$  such that  the code $\sigma(\mathcal C)$ has  generator matrix of the form
\begin{align*}
G=\left[
 \begin{matrix}
   I_h & | & A' \\
  0 & | & A''
  \end{matrix}
  \right],
\end{align*}
and $\sigma(\mathcal C) \cap (\sigma(\mathcal C))^{\perp}=Span\{G(1,:), \cdots, G(h, :)\}$, where $I_h$ is the $h \times h$ identity matrix
and $G(i,:)$ is the $i$-th row of $G$.
\end{lemma}
\begin{proof}
Firstly, note that, for any $\sigma\in \mathbf{PAut}(\mathbb F_q^n)$,
\begin{align*}
\sigma(\mathcal C) \cap (\sigma(\mathcal C))^{\perp}= \sigma(\mathcal C \cap \mathcal C^{\perp}).
\end{align*}

Since $dim_{\mathbb F_q}(\mathcal C \cap \mathcal C^{\perp})=h>0$, $\mathcal C \cap \mathcal C^{\perp}$ is an $[n,h]$ linear code. Then, there exists
a coordinate permutations $\sigma\in \mathbf{PAut}(\mathbb F_q^n)$ such that  the code $ \sigma(\mathcal C \cap \mathcal C^{\perp})$,  that is $\sigma(\mathcal C) \cap (\sigma(\mathcal C))^{\perp}£©$, has a  generator matrix of the
form $G'=[I_h | A']$, where $I_h$ is the $h \times h$ identity matrix. Thus, if $k=h$, the conclusion holds.

Assume that $k>h$. Let $\mathbf c_i=\sigma^{-1}(G'(i,:))$, where $i\in \{1, \cdots, h\}$, $G'(i,:)$ is the $i$-th row of $G'$ and $\sigma^{-1}$ is the inverse permutation of $\sigma$.
Then, $\mathbf c_1, \cdots, \mathbf c_h$ is a basis of the  linear space $\mathcal C \cap \mathcal C^{\perp}$. We can extend the linearly independent subset $\{ \mathbf c_1, \cdots, \mathbf c_h \}$  to
 a basis $\{ \mathbf c_1,   \cdots, \mathbf  c_h,  \mathbf c_{h+1}', \cdots,  \mathbf c_k'\}$ of $\mathcal C$.
 Let $\sigma(\mathbf c_i')=(c_{i,1}', \cdots, c_{i,n}')$  and $\mathbf c_i= c_i'-\sum_{j=1}^{h} c_{i,j}' \mathbf c_j$ for $i\in \{h, \cdots, k\}$.
 Then, $\{ \mathbf c_1,   \cdots, \mathbf  c_h,  \mathbf c_{h+1}, \cdots,  \mathbf c_k\}$ is also a basis of $\mathcal C$. Let $\sigma(\mathbf c_i)=(\overline{c}_{i,1}, \cdots, \overline{c}_{i,n})$ for $i\in \{1, \cdots, k\}$. If $1\le i\le h$, $\sigma(\mathbf c_i)=G'(i,:)$. If $i\in \{h+1, \cdots, k\}$ and $j\in \{1, \cdots, h\}$,
 $\overline{c}_{i,j}=c_{i,j}'-c_{i,j}'G'(j,j)=0$, where $G'(j,j)$ is the entry of the $j$-th row and the $j$-th column of $G'$.

 Hence,  $\{ \sigma(\mathbf c_1),   \cdots, \sigma(\mathbf  c_h),
 \sigma(\mathbf c_{h+1}), \cdots,  \sigma(\mathbf c_k)\}$  is a basis of $\sigma(\mathcal C)$, and, $\sigma(\mathcal C)$ has a  generator matrix of the form
\begin{align*}
G=\left[
 \begin{matrix}
   I_h & | & A' \\
  0 & | & A''
  \end{matrix}
  \right],
\end{align*}
which completes the proof.

\end{proof}
The following lemma will be important to prove the main result of this section.

\begin{lemma}\label{lem:pi-C1-C2}
Let $\pi \in \mathbf{PAut}(\mathbb F_q^n)$,  $\mathcal C_1$ and $\mathcal C_2$ be $q$-ary linear codes of length $n$. Then, $\pi(\mathcal C_1) \cap \mathcal C_2^{\perp}=\{0\}$
if and only if $\mathcal C_1 \cap (\pi^{-1}(\mathcal C_2))^{\perp}=\{0\}$.
\end{lemma}
\begin{proof}
Assume $\pi(\mathcal C_1) \cap \mathcal C_2^{\perp}=\{0\}$ and let $\mathbf c=(c_1, \cdots, c_n) \in \mathcal C_1 \cap (\pi^{-1}(\mathcal C_2))^{\perp}$.
Thus, for any $\mathbf d \in \mathcal  C_2$,
\begin{align*}
<\mathbf c, \pi^{-1}(\mathbf d)>=& \sum_{i=1}^n c_i d_{\pi^{-1}(i)}\\
=& \sum_{i=1}^n c_{\pi(i)} d_{i}\\
=& <\pi(\mathbf c), \mathbf d>=0.
\end{align*}
Then, $\pi(\mathbf c)=0$ and $\mathbf c=0$. Hence, $\mathcal C_1 \cap (\pi^{-1}(\mathcal C_2))^{\perp}=\{0\}$.

Similarly, one can show that if $\mathcal C_1 \cap (\pi^{-1}(\mathcal C_2))^{\perp}=\{0\}$, then $\pi(\mathcal C_1) \cap \mathcal C_2^{\perp}=\{0\}$.

\end{proof}
We are now in position to present the main result of this section proving the existence of a certain mapping $\sigma$ for which a linear code $\mathcal  C$ can be viewed as a $\sigma$-LCD code.

\begin{theorem}\label{thm:sigma-LCP}
Let $\mathcal  C$ be a linear code in $\mathbb F_q^n$. Then,

(i) if $q>2$ then there exists a mapping $\sigma\in \mathbf{MAut}(\mathbb F_q^n)$ such that $\mathcal C$ is $\sigma$-LCD;

(ii) if $q=2$ then there exists a mapping $\sigma\in \mathbf{PAut}(\mathbb F_2^{n+1})$ such that $\{0\} \times \mathcal C:=\{(0,\mathbf c) \in \mathbb F_2 \times \mathbb F_2^n: \mathbf c\in \mathcal C\}$ is $\sigma$-LCD.
\end{theorem}
\begin{proof}
Let $h=dim_{\mathbb F_q}(\mathcal C \cap \mathcal C^{\perp})$. If $h=0$, the conclusion holds by choosing $\sigma$ the identity map.

Assume that  $h>0$.
From  Lemma \ref{lem:hull}, there  exists
$\pi \in \mathbf{PAut}(\mathbb F_q^n)$ such that  the code $\pi (\mathcal C)$ has  generator matrix of the form
\begin{align*}
G=\left[
 \begin{matrix}
   I_h & | & A' \\
  0 & | & A''
  \end{matrix}
  \right],
\end{align*}
and $\pi(\mathcal C) \cap (\pi (\mathcal C))^{\perp}=Span\{G(1,:), \cdots, G(h, :)\}$. Thus,
\begin{align*}
G G^{T}=&
\left[
 \begin{matrix}
   I_h  & A' \\
  0   & A''
  \end{matrix}
  \right]
  \left[
 \begin{matrix}
    I_h  & A' \\
  0   & A''
  \end{matrix}
  \right]^{T}\\
  =&\left[
 \begin{matrix}
    I_h +A' A'^{T}  & A' A''^{T} \\
  A'' A'^{T}   & A'' A''^{T}
  \end{matrix}
  \right]\\
  =&\left[
 \begin{matrix}
    0  & 0 \\
  0   & A'' A''^{T}
  \end{matrix}
  \right].
\end{align*}
One gets
\begin{align}\label{eq:A-A}
A' A'^{T}=-I_h, ~~A' A''^{T}=0,~~ A'' A'^{T}=0.
\end{align}
By Proposition \ref{prop:dim-hull}, one obtains
\begin{align}\label{eq:rank-AA}
Rank(A'' A''^{T})=Rank(GG^T)=k-h.
\end{align}

(i) If $q>2$, there exists $\lambda\in \mathbb F_q $ such that $\lambda \not \in \{0, 1\}$.
Let  $\gamma$ be the linear  transformation defined by $\gamma(\mathbf c)= (w_1, \cdots,w_n)$
for any $\mathbf c=(c_1, \cdots, c_n)\in \mathbb F_q^n$, where $w_i=\lambda c_i$ if $1\le i\le h$ and  $w_i= c_i$ if $h+1\le i\le n$.
Then,
\begin{align*}
G (G^{\gamma})^{T}=&
\left[
 \begin{matrix}
   I_h  & A' \\
  0   & A''
  \end{matrix}
  \right]
  \left[
 \begin{matrix}
   \lambda I_h  & A' \\
  0   & A''
  \end{matrix}
  \right]^{T}\\
  =&\left[
 \begin{matrix}
   \lambda I_h +A' A'^{T}  & A' A''^{T} \\
  A'' A'^{T}   & A'' A''^{T}
  \end{matrix}
  \right].
\end{align*}
From Equation (\ref{eq:A-A}), one has
\begin{align*}
G (G^{\gamma})^{T}  =\left[
 \begin{matrix}
   (\lambda-1) I_h   & 0 \\
  0   & A'' A''^{T}
  \end{matrix}
  \right].
\end{align*}
From Equation (\ref{eq:rank-AA}), $Rank(G (G^{\gamma})^{T})=k$. By Proposition \ref{prop:dim-hull},
$\pi(\mathcal C) \cap (\pi(\mathcal C))^{\perp_{\gamma}}=\{0\}$.
Using Equation (\ref{eq:sigma-euclidean-dual}), one has
\begin{align*}
\pi(\mathcal C) \cap (\gamma\pi(\mathcal C))^{\perp}=\{0\}.
\end{align*}
By Lemma \ref{lem:pi-C1-C2}, $\mathcal C \cap (\pi^{-1}\gamma\pi(\mathcal C))^{\perp}=\{0\}$. From Equation (\ref{eq:sigma-euclidean-dual}),
$\mathcal C$ is $\sigma$-LCD, where $\sigma= \pi^{-1}\gamma\pi\in \mathbf{MAut}(\mathbb F_q^n)$.

(ii) Let $q=2$. Then, the code $\{0\}\times \pi (\mathcal C)$ has a  generator matrix of the form
\begin{align*}
G=\left[
 \begin{matrix}
  0& | & I_h & | & A' \\
  0& | &0 & | & A''
  \end{matrix}
  \right].
\end{align*}
Let  $\pi_2$ be the linear  transformation defined by $\pi_2(\mathbf c)= (w_0, w_1, \cdots,w_n)$
for any  $\mathbf c=(c_0, c_1, \cdots, c_n)\in \mathbb F_q^{n+1}$, where $w_h=c_0$,  $w_i= c_{i+1}$ if $0\le i\le h-1$, and  $w_i= c_i$ if $h+1\le i\le n$.
Then  $\pi_2(\mathbf c)= (c_1,c_2, \cdots, c_{h},c_{0}, c_{h+1}, c_{h+2} \cdots,c_n)$ and
\begin{align*}
G (G^{\pi_2})^{T}=&
\left[
 \begin{matrix}
   0 & I_h  & A' \\
  0 & 0   & A''
  \end{matrix}
  \right]
  \left[
 \begin{matrix}
    I_h& 0  & A' \\
  0 & 0  & A''
  \end{matrix}
  \right]^{T}\\
  =&\left[
 \begin{matrix}
   J_h +A' A'^{T}  & A' A''^{T} \\
  A'' A'^{T}   & A'' A''^{T}
  \end{matrix}
  \right],
\end{align*}
where
\begin{align*}
J_h=\left[
 \begin{matrix}
   0 & 1 & 0 & \cdots & 0 & 0 \\
  0 & 0  &1 & \cdots & 0  & 0  \\
  \vdots & \vdots  &\vdots & \cdots & \vdots  & \vdots \\
   0 & 0  &0 & \cdots & 1  & 0  \\
    0 & 0  &0 & \cdots & 0  & 1  \\
     0 & 0  &0 & \cdots & 0  & 0
  \end{matrix}
  \right].
\end{align*}

From Equation (\ref{eq:A-A}), one has
\begin{align*}
G (G^{\pi_2})^{T}  =\left[
 \begin{matrix}
    I_h+J_h   & 0 \\
  0   & A'' A''^{T}
  \end{matrix}
  \right].
\end{align*}
From Equation (\ref{eq:rank-AA}), $Rank(G (G^{\pi_2})^{T})=k$. By Proposition \ref{prop:dim-hull},
$(\{0\} \times \pi(\mathcal C)) \cap (\{0\} \times \pi(\mathcal C))^{\perp_{\pi_2}}=\{0\}$.
Let $\pi_1\in \mathbf{PAut}(\mathbb F_q^{n+1})$ be defined by $\pi_1(c_0,\mathbf c)=(c_0, \pi(\mathbf c))$, where $c_0\in \mathbb F_q$ and  $\mathbf c\in \mathbb F_q^n$.
Thus, $\{0\}\times  \pi (\mathcal C)=  \pi_1(\{0\}\times \mathcal C)$.
Using Equation (\ref{eq:sigma-euclidean-dual}), one has
\begin{align*}
\pi_1(\{0\}\times \mathcal C) \cap (\pi_2\pi_1(\{0\}\times\mathcal C))^{\perp}=\{0\}.
\end{align*}
By Lemma \ref{lem:pi-C1-C2}, $\{0\}\times \mathcal C \cap (\pi_1^{-1}\pi_2\pi_1(\{0\}\times \mathcal C))^{\perp}=\{0\}$. From Equation (\ref{eq:sigma-euclidean-dual}),
$\{0\}\times \mathcal C$ is $\sigma$-LCD, where $\sigma= \pi_1^{-1}\pi_2\pi_1 \in \mathbf{PAut}(\mathbb F_q^{n+1})$.

\end{proof}


\begin{remark}
A vector $(x_1, \cdots,x_n)\in\mathbb F_q^n$ is called \emph{even-like} if $\sum_{i=1}^n x_i=0$. A code is said to be even-like if it has only even-like codewords.
Let $n$ be an even positive integer and $\mathcal C$ be an even-like linear code in $\mathbb F_2^n$ with $(1,1, \cdots, 1)\in \mathcal C$. Since  $(1,1, \cdots, 1) \in \mathcal  C \cap (\mathcal C)^{\perp_{\sigma}}$ for any $\sigma\in \mathbf{PAut}(\mathbb F_2^n)$
, $\mathcal C$
cannot be a $\sigma$-LCD.
\end{remark}

\begin{remark}
Let $\mathcal C_1$ and $\mathcal C_2$ be $q$-ary $[n,k]$ linear codes. Using similar  techniques as in Theorem \ref{thm:sigma-LCP}, we can actually prove the more general conclusions below:

(i) if $q>2$ then there exists $\sigma\in \mathbf{MAut}(\mathbb F_q^n)$ such that $(\mathcal C_1, (\sigma(\mathcal C_2))^\perp)$ is an LCP of codes;

(ii) if $q=2$ then there exists $\sigma\in \mathbf{PAut}(\mathbb F_q^{n+1})$ such that $(\{0\}\times \mathcal C_1, (\sigma(\{0\}\times \mathcal C_2))^\perp)$ is an LCP of codes.
\end{remark}
Consequently,  in order to  construct effectively LCP of codes,  the most important task is to determine the specific $\sigma$ inner product such that $\mathcal C_1$ is a $\sigma$-LCD,
or $(\mathcal C_1, (\sigma(\mathcal C_2))^\perp)$ is LCP of codes, or
$(\{0\}\times \mathcal C_1, (\sigma(\{0\}\times \mathcal C_2))^\perp)$ is LCP of codes. In fact, the proofs of   Lemma \ref{lem:hull} and  Theorem \ref{thm:sigma-LCP} also show us
 how to construct  effectively $\sigma$ such that the code is $\sigma$-LCD. Further, in Section \ref{sec:LCD-GQC} and \ref{sec:Abelian-LCD},
we will focus primarily on describing more explicitly $\sigma$-LCD generalized quasi-cyclic codes and Abelian codes.
We highlight that from Theorem \ref{thm:sigma-LCP}
and the previous remarks,  the following two corollaries hold.

\begin{corollary}\label{cor:q>2}
Let $q>2$. Given two  $q$-ary $[n, k,d_i]$ ($i=1,2$) linear codes,  there exists
an  LCP of codes with parameters  $[n, k, d_1, d_2]$.
\end{corollary}
\begin{proof}
Let $\mathcal C_i$ be an $[n,k, d_i]$ linear codes for $i=1,2$. From the remarks after Theorem \ref{thm:sigma-LCP}, there exists $\sigma\in \mathbf{MAut}(\mathbb F_q^n)$ such that $(\mathcal C_1, (\sigma(\mathcal C_2))^\perp)$ is an LCP of codes with parameters $[n, k, d_1, d_2]$.
\end{proof}

\begin{remark}
We let $\mathcal G_{24}$ be the binary Golay codes with parameters $[24, 12, 8]$. As we all know that $\mathcal G_{24}$ is self-dual and $(1,1, \cdots, 1)\in \mathcal G_{24}$.
Moreover, any linear code $\mathcal C$ with parameters $[24, 12, 8]$ must be  equivalent to $\mathcal G_{24}$. Thus, $(1,1, \cdots, 1)\in \mathcal C$ and $\mathcal C$ is self-dual.
It shows that there is no binary LCP of codes with parameters $[24, 12, 8, 8]$. On the contrary, we assume there exists binary LCP of codes $(\mathcal C_1, C_2)$ with parameters $[24, 12, 8, 8]$. Then, $\mathcal C_1+ \mathcal C_2= \mathbb F_2^{24}$, $\mathcal C_1 \cap \mathcal C_2=\{0\}$ and $\mathcal C_2^{\perp}$ is an $[24,12,8]$ linear code. Thus, $\mathcal C_2^{\perp}=\mathcal C_2$ and $(1,1, \cdots, 1)\in \mathcal C_1 \cap \mathcal C_2$, which contradicts $\mathcal C_1 \cap \mathcal C_2=\{0\}$. The above discussion shows that when $q=2$, the Corollary
\ref{cor:q>2} is not true. For binary LCP of codes, we have the following slightly weaker conclusions.
\end{remark}

\begin{corollary}
Given two binary $[n, k,d_i]$ ($i=1,2$) linear codes, there  exist binary  LCP of codes with parameters $[n+1, k, d_1, d_2]$ and  $[n, k, \ge (d_1-1), \ge (d_2-1)]$.
\end{corollary}
\begin{proof}
If $n=k$, The conclusion is clear. Next, we assume $k<n$.
Let $\mathcal C_i$ be a binary $[n,k, d_i]$ linear codes for $i=1,2$. From the remarks after Theorem \ref{thm:sigma-LCP}, there exists $\sigma\in \mathbf{PAut}(\mathbb F_q^{n+1})$ such that $(\{0\}\times \mathcal C_1, (\sigma(\{0\}\times \mathcal C_2))^\perp)$ is an LCP of codes with parameters $[n+1, k, d_1, d_2]$.

If $k<n$ and there exist $[n, k,d_i]$  linear codes $\mathcal C_i$ ($i=1,2$), there must exist $[n-1, k,\ge (d_i-1])$  linear codes $\mathcal C_i^*$ ($i=1,2$), which are obtained by deleting some same coordinate  in each codeword of $\mathcal C_i$. Using the previous discussion,  there  exists a binary  LCP of code with parameters   $[n, k, \ge (d_1-1), \ge (d_2-1)]$.
It completes the proof.
\end{proof}

\section{$\sigma$-LCD  generalized quasi-cyclic codes}\label{sec:LCD-GQC}
In \cite{GOOSSS17}, G\"{u}neri et al. presented a criterion for an Euclidean GQC code to be LCD code and analyzed the asymptotic performance of Euclidean LCD GQC codes.
In this section, we shall study $\sigma$-LCD GQC codes and present results on those codes. More specifically, we present necessary and sufficient conditions for
 GQC codes and $1$-generator GQC codes to be $\sigma$-LCD codes
and derive asymptotically good $\sigma$-LCD GQC codes.

\subsection{$\sigma$-LCD GQC codes}
Let $m_1,m_2,\cdots,m_l$ and $n$ be positive integers with $n=m_1+m_2+\cdots + m_{l}$. Let $R_i = \mathbb F_q[x]/(x^{m_i}-1)$ and $\mathbf{gcd}(m_i,q) = 1$, where $1 \le i \le l$.
The cartesian product $M = R_1 \times R_2 \times \cdots \times R_l$ is an $\mathbb F_q[x]$-module under the component wise addition and the scalar multiplication.
If $\mathcal C$ is an $\mathbb F_q[x]$-submodule, then  $\mathcal C$ is called a \emph{generalized  quasi-cyclic} (\emph{GQC}) code of block lengths $(m_1,m_2,\cdots,m_l)$,
which is a $q$-ary linear code of length $n$.

If $m_1= \cdots = m_{l}= m$, then we obtain a \emph{quasi-cyclic} (\emph{QC}) code of length $ml$ and index $l$. Furthermore, if $l = 1$ then $\mathcal C$ is a cyclic code of length $m$.

Let $m$ be  the least common multiple of $m_1, m_2, \cdots, m_{l}$ and  $\xi$ be a primitive $m$-th root of unity in the algebraic closure of $\mathbb F_q$.
The \emph{$q$-cyclotomic coset} $\mathbb C_{i}$ modulo $m$ containing $i$ is defined by
\begin{align*}
\mathbb C_{i} = \{iq^0,iq^1,iq^2,\cdots ,iq^{l_i-1}\} \pmod m,
\end{align*}
where $l_i$ is the least positive integer such that $i q^{l_i} \equiv i \pmod m$. The smallest non-negative integer in $\mathbb C_i$ is called the \emph{coset leader} of $\mathbb C_i$.
Let $\Gamma(q,m)$ denotes the set of all coset leaders of the $q$-cyclotomic cosets modulo $m$. Then $\{\mathbb C_i : i \in \Gamma(q,m)\}$ forms a partition of the set
$\mathbb Z_m = \{0,1,\cdots ,m-1\}$. Clearly, the minimal polynomial $\mathbb M_{\xi^i} (x)$ over $\mathbb F_q$ of $\xi^i$ is given by
\begin{align}\label{eq:M-xi-i}
\mathbb M_{\xi^i}(x)= \prod_{j\in \mathbb C_i} (x-\xi^j).
\end{align}

Set $\Gamma_0:=\{i\in \Gamma(q,m) : -i \in \mathbb C_i\}$. Then, there exists a set $\Gamma_1\subseteq \Gamma(q,m)$ such that
$(\cup_{i\in \Gamma_1} \mathbb C_i)\cap (\cup_{i\in \Gamma_1}\mathbb C_{-i})=\emptyset$ and $(\cup_{i\in \Gamma_0} \mathbb C_i)\cup (\cup_{i\in \Gamma_1} \mathbb C_i)\cup (\cup_{i\in \Gamma_1}\mathbb C_{-i})=\mathbb Z_m$. Let $\Gamma_{0,+}:=\{+1,-1\}$ when $q$ is odd and $m$ is even, and $\Gamma_{0,+}:=\{+1\}$ otherwise. Set
$\Gamma_{0,-}:=\Gamma_{0}\setminus \Gamma_{0,+}$.
Clearly,
\begin{align*}
x^m-1= \prod_{i\in \Gamma_{0,+}} \mathbb M_{\xi^i}(x)  \prod_{i\in \Gamma_{0,-}} \mathbb M_{\xi^i}(x)  \prod_{i\in \Gamma_1} \mathbb M_{\xi^i}(x) \mathbb M_{\xi^{-i}}(x).
\end{align*}
Therefore, for every $j\in\mathbb C_i$ one has:
\begin{align*}
x^{m_j}-1=\prod_{i\in \Gamma_{0,+}} \mathbb M_{\xi^i}^{\delta_{j,i}}(x)  \prod_{i\in \Gamma_{0,-}} \mathbb M_{\xi^i}^{\delta_{j,i}}(x)   \prod_{i\in \Gamma_1} \left( \mathbb M_{\xi^i}(x) \mathbb M_{\xi^{-i}}(x) \right )^{\delta_{j,i}},
\end{align*}
where $\delta_{j,i}=1$ if $\xi^{im_j}=1$ and $\delta_{j,i}=0$ otherwise.

For $i \in \Gamma(q,m)$, we now define
 \begin{align*}
 \mathbb V_i:= \{(y_1,y_2, \cdots, y_{l})\in \mathbb F_q[\xi^i]^l: y_j=0 \text{ if } \xi^{im_j}\neq 1\}.
 \end{align*}
By definition,  $\mathbb V_i=\{0\}$ if and only if  $\mathbb V_{ai}=\{0\}$, where  $\mathbf{gcd}(a,m)=1$. Moreover, $\mathbb V_{ai}=\mathbb V_{i}$ when $\mathbf{gcd}(a,m)=1$.

Let $\mathcal C$ be a GQC code. For $i\in \mathbb Z_m$,  we define the  \emph{constituent} $C_i$ as follows.
\begin{align}\label{eq:constituent}
C_i=\{(c_1(\xi^i) \delta_{1,i}, c_2(\xi^i) \delta_{2,i}, \cdots, c_l(\xi^i) \delta_{l,i}): (c_1(x), c_2(x), \cdots, c_l(x))\in \mathcal C\},
\end{align}
where the symbol $\delta_{j,i}$ is defined as above. For $i\in \mathbb Z_m$,  $C_i$ is clearly  an $\mathbb F_q[\xi^i]$ linear subspace of $\mathbb V_i$.
Moreover, for any integer $k$,
\begin{align}\label{eq:C-iq}
C_{iq^k}=C_{i}^{q^k}:= \{(\alpha_1^{q^k}, \cdots, \alpha_l^{q^k}): (\alpha_1, \cdots, \alpha_l)\in \mathcal C_i\}.
\end{align}

Given an integer $a$ and a subset $A$ of $\mathbb Z_m$, the set $\{a x\pmod{m} : x\in A\}$ is denoted by $aA$. From \cite{GOOSSS17}, a GQC code $\mathcal C$ can be viewed as
\begin{align}\label{eq:decomp}
\mathcal C\cong \left(\prod_{i\in \Gamma_{0,+}} C_i \right)\times \left( \prod_{i\in \Gamma_{0,-}} C_i\right)\times   \left (\prod_{i\in \Gamma_1\cup (-\Gamma_1)} C_i  \right).
\end{align}
It is well-known, that the dual of a GQC code is also a GQC code. Let $\mathcal  C$ be a GQC code with  decomposition given by Equation (\ref{eq:decomp}). Then its dual code $\mathcal C^{\perp}$ is of the form
\begin{align}\label{eq:C-dual}
\mathcal C^{\perp}\cong \left(\prod_{i\in \Gamma_{0,+}} C_i^{\perp'} \right)\times \left( \prod_{i\in \Gamma_{0,-}} C_i^{\perp_{h}'}\right)\times   \left (\prod_{i\in \Gamma_1\cup (-\Gamma_1)} C_{-i}^{\perp'}  \right),
\end{align}
where $C_i^{\perp'}$ is the  \emph{Euclidean $\mathbb V_i$-dual} of $C_i$  defined by
\begin{align*}
C_i^{\perp'}=\{(w_1,\cdots, w_{l})\in \mathbb V_i: \sum_{j=1}^lc_jw_j=0 \text{ for any } (c_1,\cdots, c_l)\in C_i\},
\end{align*}
and
$C_i^{\perp'_h}$ is the  \emph{Hermitian $\mathbb V_i$-dual} of $C_i$  defined by
\begin{align*}
C_i^{\perp'_h}=\{(w_1,\cdots, w_{l})\in \mathbb V_i: \sum_{j=1}^lc_jw_j^{q^{deg(\mathbb M_{\xi^i}(x))/2}}=0 \text{ for any } (c_1,\cdots, c_l)\in C_i\}.
\end{align*}
Note that if $i \in \Gamma_{0,-}$, then $(\xi^{i})^{-1}=(\xi^{i})^{ deg(\mathbb M_{\xi^i}(x))/2}$. Consequently,
\begin{align*}
C_i^{\perp'_h}=C_{-i}^{\perp'}.
\end{align*}

Let $a$ be an integer with $\mathbf{gcd}(a,m)=1$. Let $\mu_a$ be the transformation on $M = R_1 \times R_2 \times \cdots \times R_l$ defined by
\begin{align*}\mu_a(c_1(x), c_2(x), \cdots, c_l(x))=(w_1(x),w_2(x), \cdots, w_l(x)),\end{align*}
where $(c_1(x), c_2(x), \cdots, c_l(x))\in M$ and $w_j(x)=c_j(x^a) \pmod{ (x^{m_j}-1)}$ for $j\in \{1,2,\cdots, l\}$.
Then $\mu_a$ is a coordinate permutation over $M$. Let $\mathcal C'=\mu_a(\mathcal C)$ and $C_i'$ ($i\in \mathbb Z_m$) be the constituents of $\mathcal C'$.
Then,
\begin{align*}
C_i'&=\{(c_1(\xi^{ai}) \delta_{1,i},  \cdots, c_l(\xi^{ai}) \delta_{l,i}): (c_1(x), \cdots, c_l(x))\in \mathcal C\}\\
&=\{(c_1(\xi^{ai}) \delta_{1,ai},  \cdots, c_l(\xi^{ai}) \delta_{l,ai}): (c_1(x), \cdots, c_l(x))\in \mathcal C\}=C_{ai}.\\
\end{align*}
Note that $ai\equiv i \pmod{m}$ if $i\in \Gamma_{0,+}$. Thus, $\mu_a(\mathcal C)$
can be decomposed as
\begin{align}\label{eq:mu-C}
\mu_a(\mathcal C)\cong \left(\prod_{i\in \Gamma_{0,+}} C_{i} \right)\times \left( \prod_{i\in \Gamma_{0,-}} C_{ai}\right)\times   \left (\prod_{i\in \Gamma_1\cup (-\Gamma_1)} C_{ai}  \right).
\end{align}
One observes that $a\Gamma_{0,+}=\Gamma_{0,+}$, $a\Gamma_{0,-}=\Gamma_{0,-}$, and $a(\Gamma_{1}\cup (- \Gamma_{1}))=\Gamma_{1}\cup (- \Gamma_{1})$. By Equation (\ref{eq:C-dual}), one has
\begin{align}\label{eq:mu-C-dual}
(\mu_a(\mathcal C))^{\perp}\cong \left(\prod_{i\in \Gamma_{0,+}} C_{i}^{\perp'} \right)\times \left( \prod_{i\in \Gamma_{0,-}} C_{ai}^{\perp_{h}'}\right)\times   \left (\prod_{i\in \Gamma_1\cup (-\Gamma_1)} C_{-ai}^{\perp'}  \right).
\end{align}
From Equation (\ref{eq:-i-Frob}) and $-ai\equiv i \pmod{m}$ if $i \in \Gamma_{0,+}$, one gets
\begin{align}\label{eq:mu-C-dual-simple}
(\mu_a(\mathcal C))^{\perp}\cong \prod_{i\in \Gamma(q,m)} C_{-ai}^{\perp'}.
\end{align}

Let us now characterize $\mu_a$ self-orthogonal, $\mu_a$ self-dual, and $\mu_a$-LCD GQC codes via their constituents. Note that the Euclidean case (that is, when $a=1$) has been already considered in \cite{GZS16} and \cite{GOOSSS17}.

\begin{theorem}\label{thm:mu-a}
Let $a$ be an integer with $\mathbf{gcd}(a,m)=1$ and $\mathcal C$ be a $q$-ary GQC code of block lengths $(m_1, \cdots,  m_{l})$, whose  decomposition is as in Equation (\ref{eq:decomp}). Then

(i) $\mathcal C$ is $\mu_a$  self-orthogonal (resp.  $\mu_a$ self-dual) if and only if $C_i\subseteq C_{i}^{\perp '}$ (resp. $C_i= C_{i}^{\perp '}$) for $i\in \Gamma_{0,+}$,
$C_i\subseteq C_{ai}^{\perp_h '}$ (resp. $C_i= C_{ai}^{\perp_h '}$) for $i\in \Gamma_{0,-}$, and
$C_i\subseteq C_{-ai}^{\perp '}$ (resp. $C_i= C_{-ai}^{\perp '}$) for $i\in \Gamma_1\cup(-\Gamma_1)$;

(ii) $\mathcal C$ is $\mu_a$-LCD if and only if $C_i \cap  C_{i}^{\perp '}=\emptyset$ for $i\in \Gamma_{0,+}$,
$C_i\cap C_{ai}^{\perp_h '}=\emptyset$ for $i\in \Gamma_{0,-}$, and
$C_i\cap C_{-ai}^{\perp '}=\emptyset$ for $i\in \Gamma_1\cup(-\Gamma_1)$.
\end{theorem}
\begin{proof}
 Note that $\mathcal C^{\perp_{\mu_a}}= (\mu_a(\mathcal C))^{\perp}$. Then, the conclusion follows from Equations (\ref{eq:decomp}) and
 (\ref{eq:mu-C-dual}).
\end{proof}

\begin{corollary}\label{cor:mu-a}
Let $a$ be an integer with $\mathbf{gcd}(a,m)=1$ and $\mathcal C$ be a $q$-ary GQC code of block lengths $(m_1, \cdots,  m_{l})$, whose  decomposition is given by Equation (\ref{eq:decomp}).  Then

(i) $\mathcal C$ is $\mu_a$  self-orthogonal (resp.  $\mu_a$ self-dual) if and only if
$C_i\subseteq C_{-ai}^{\perp '}$ (resp. $C_i= C_{-ai}^{\perp '}$) for $i\in \Gamma(q,m)$;

(ii) $\mathcal C$ is $\mu_a$-LCD if and only if
$C_i\cap C_{-ai}^{\perp '}=\emptyset$ for $i\in \Gamma(q,m)$.
\end{corollary}
\begin{proof}
Note that $C_{-ai}= C_{i}$ if $i\in \Gamma_{0,+}$, and $C_{-ai}^{\perp'}= C_{ai}^{\perp_{h}'}$ if $ i\in \Gamma_{0,-}$. Then, the corollary follows from Theorem \ref{thm:mu-a}.
\end{proof}
From Equation (\ref{eq:C-iq}), one gets
\begin{align*}
C_{iq^k}^{\perp'}=(C_{i}^{\perp'})^{q^k},
\end{align*}
where $k$ is a nonnegative integer.
Then,
\begin{align}\label{eq:i-iq}
C_{iq^k}\cap C_{-aiq^k}^{\perp '}=(C_i\cap C_{-ai}^{\perp '})^{q^k}.
\end{align}

From Equation (\ref{eq:i-iq}) and  Corollary \ref{cor:mu-a}, we deduce the following characterizations of $\mu_a$ self-orthogonal, $\mu_a$ self-dual and $\mu_a$-LCD, respectively.

\begin{corollary}\label{cor:mu-a-iq}
Let $a$ be an integer with $\mathbf{gcd}(a,m)=1$ and $\mathcal C$ be a $q$-ary GQC code of block lengths $(m_1, \cdots,  m_{l})$, whose  decomposition is as in Equation (\ref{eq:decomp}). Then

(i) $\mathcal C$ is $\mu_a$  self-orthogonal (resp.  $\mu_a$ self-dual) if and only if
$C_i\subseteq C_{-ai}^{\perp '}$ (resp. $C_i= C_{-ai}^{\perp '}$)  for $i\in \mathbb Z_{m}$;

(ii) $\mathcal C$ is $\mu_a$-LCD if and only if
$C_i\cap C_{-ai}^{\perp '}=\emptyset$  for $i\in \mathbb Z_m$.
\end{corollary}

By choosing $a=-1$ in Corollary \ref{cor:mu-a}, one has

\begin{corollary}\label{cor:mu--1}
Let  $\mathcal C$ be a $q$-ary GQC code of block lengths $(m_1, \cdots,  m_{l})$, whose  decomposition is as in Equation (\ref{eq:decomp}). Then

(i) $\mathcal C$ is $\mu_{-1}$  self-orthogonal (resp.  $\mu_{-1}$ self-dual) if and only if
$C_i\subseteq C_{i}^{\perp '}$ (resp. $C_i= C_{i}^{\perp '}$)  for $i\in \Gamma(q,m)$;

(ii) $\mathcal C$ is $\mu_{-1}$-LCD if and only if
$C_i\cap C_{i}^{\perp '}=\emptyset$ for $i\in \Gamma(q,m)$.
\end{corollary}

The following result characterizes $\mu_a$-LCD GQC codes.

\begin{theorem}\label{thm:mu-a-0-Vi}
Let $a$ be an integer with $\mathbf{gcd}(a,m)=1$ and $\mathcal C$ be a $q$-ary GQC code of block lengths $(m_1, \cdots,  m_{l})$ with constituent $C_i=\{0\} \text{ or } \mathbb V_i$ for $i\in \Gamma(q,m)$. Let $T$ be the set of all $i\in \mathbb Z_m$ with $C_i=\{0\}$. Then $\mathcal C$ is a $\mu_a$-LCD GQC code if and only if one of the following statements holds:

(i) $S=-aS$, where $S=\mathbb Z_m \setminus T$ and $-aS=\{-as \pmod{m}: s\in S\}$;

(ii) $T=-aT$;

(iii) $\mu_{-a}(\mathcal  C)=\mathcal C$.
\end{theorem}
\begin{proof}
From Equations (\ref{eq:decomp}) and (\ref{eq:mu-C-dual-simple}), $\mathcal C$ is $\mu_a$-LCD if and only if $\mathcal C_{i} \cap \mathcal C_{-ai}=\{0\}$. By the definition of $S$
and $T$, $\mathcal C$ is $\mu_a$-LCD if and only if $S=-aS$. Thus, $\mathcal C$ is $\mu_a$ LCD if and only if $T=-aT$.

Let  $\mathcal C$ be $\mu_a$-LCD. From Equation (\ref{eq:decomp}), one has
\begin{align*}
\mathcal C \cong  \left( \prod_{i\in \Gamma(q,m) \cap T}  \{0\} \right) \times  \left (\prod_{i\in \Gamma(q,m) \cap S} \mathbb V_i\right).
\end{align*}
 Note that $T=-aT$, $S=-a S$, and $\mathbb V_{-ai}=\mathbb V_i$. By  Equation (\ref{eq:mu-C}) we obtain
\begin{align*}
\mu_{-a}(\mathcal C)\cong \left( \prod_{i\in \Gamma(q,m) \cap T} \{0\} \right) \times   \left( \prod_{i\in \Gamma(q,m) \cap S} \mathbb V_i \right).
\end{align*}
Hence, $\mu_{-a}(\mathcal  C)=\mathcal C$.

Conversely, let $\mu_{-a}(\mathcal  C)=\mathcal C$. From Equation (\ref{eq:mu-C}), one has
\begin{align*}
\mu_{-a} (\mathcal C) &\cong  \left( \prod_{i\in \Gamma(q,m)\cap T}  C_{-ai} \right) \times  \left( \prod_{i\in \Gamma(q,m)\cap S}  C_{-ai} \right)\\
& = \left( \prod_{i\in \Gamma(q,m)\cap T} \{0\} \right) \times  \left( \prod_{i\in \Gamma(q,m)\cap S}  \mathbb V_i \right)  (\cong \mathcal C).
\end{align*}

Since $\mathbb V_i\neq \{0\}$ if $i\in S$ and $\mathbb V_{-ai}=\mathbb V_{i}$, one obtains $-aS=S$. Hence, $\mathcal  C$ is $\mu_a$-LCD from the above discussion.

\end{proof}
The following result shows that GQC codes can be viewed as $\sigma$-LCD codes for an appropriate mapping $\sigma$.
\begin{corollary}\label{cor:0-Vi-LCD}
Let  $\mathcal C$ be a $q$-ary GQC code of block lengths $(m_1, \cdots,  m_{l})$ with $\mathbf{gcd}(\prod_{j=1}^l m_j,q)=1$ and $C_i$ ($i\in \Gamma(q,m)$) be its constituents defined as Equation (\ref{eq:constituent}).
If $C_i=\{0\} \text{ or } \mathbb V_i$ for $i\in \Gamma(q,m)$, then $\mathcal  C$ is always a $\mu_{-q^j}$-LCD code for any nonnegative integer $j$. In particular, any cyclic code  $\mathcal C$ of length $m$ with $\mathbf{gcd}(q,m)=1$ is $\mu_{-1}$-LCD and $\sigma$-LCD, where $\sigma$
is defined by
$$\sigma(c_0, c_1,\cdots, c_{m-1})=(c_{m-1}, c_{m-2}, \cdots, c_0),$$
for $(c_0, c_1,\cdots, c_{m-1})\in \mathcal C$.
\end{corollary}
\begin{proof}
Let $S=\{i\in \mathbb Z_m: C_i\neq \{0\}\}$. For any $i \in S$, $ C_{iq^j}=  C_{i}^{q^j}\neq \{0\}$. Thenfore $-(-q^j)i\in S$, which implies that  $-(-q^j)S=S$. From Theorem \ref{thm:mu-a-0-Vi}, we deduce that $\mathcal C$ is $\mu_{-q^j}$-LCD.

 Now, let $\mathcal C$ be a cyclic code. In this case, $\mathbb V_i=\mathbb F_q[\xi^i]$ and $C_i=\{0\} \text{ or } \mathbb V_i$.  Therefore, $\mathcal  C$ is $\mu_{-1}$-LCD,
 that is, $\mathcal  C \cap \mathcal C^{\perp_{\mu_{-1}}}=\{0\}$. For any $c_0+c_1x+ \cdots +c_{m-1} x^{m-1}$, $\mu_{-1}(c_0+c_1x+ \cdots + c_{m-1} x^{m-1}) = c_0+c_{m-1}x
 +c_{m-2} x^2+\cdots +c_2 x^{m-2}+c_1x^{m-1}$. Then,
 \begin{align*}
 \mu_{-1}(\mathcal C)&=\{c_0+c_{m-1}x
 +\cdots +c_1x^{m-1}: c_0+c_1x+ \cdots+ c_{m-1} x^{m-1}\in \mathcal  C\}\\
 &=\{x^{m-1}(c_0+c_{m-1}x
 +\cdots +c_1x^{m-1}): c_0+c_1x+ \cdots +c_{m-1} x^{m-1}\in \mathcal  C\}\\
 &= \{c_{m-1}+\cdots +c_1x^{m-2}+c_0x^{m-1}: c_0+c_1x+ \cdots +c_{m-1} x^{m-1}\in \mathcal  C\},
 \end{align*}
where the second identity  follows from the fact that $\mu_{-1}(\mathcal C)$ is a cyclic code.

 Thus, $\mathcal C^{\perp_{\mu_{-1}}}=(\mu_{-1}(C))^{\perp}=(\sigma(\mathcal C))^{\perp}$ and
 $\mathcal  C \cap \mathcal C^{\perp_{\sigma}}=\{0\}$. Hence, $\mathcal C$ is $\sigma$-LCD.
\end{proof}

Let $\mathcal C$ be a GQC code of block lengths $(m_1, \cdots, m_l)$. For any $j\in \{1,\cdots, l \}$, define
$$\mathcal C_j=\{c_j(x): (c_1(x), \cdots, c_j(x), \cdots, c_l(x))\in \mathcal C\}.$$ Then, $\mathcal  C_j$ is a cyclic code over $\mathbb F_q$ of length $m_j$. The following result shows how one can construct a $\sigma$-LCD GQC code $\mathcal C$ from $\sigma$-LCD cyclic codes  $\mathcal C_j$.

\begin{theorem}\label{thm:mi-mj-lcd}
Let $m_1, \cdots, m_l$ be positive integers with $\mathbf{gcd}(m_j,m_k)=1$ for $1\le j < k\le l$ and $q$ be a prime power  such that $\mathbf{gcd}(\prod_{j=1}^{l} m_j, q)=1$.
Let $a$ be an integer with  $\mathbf{gcd}(\prod_{j=1}^{l} m_j, a)=1$.
Then, $\mathcal C$ is a $\mu_{a}$-LCD GQC codes  if and only if $\mathcal C_j$ is a $\mu_{a}$-LCD cyclic code for any $j\in \{1,2,\cdots,l\}$ and $C_0=\{(c_1(1), \cdots,
c_l(1)): (c_1(x), \cdots, c_l(x)) \in \prod_{j=1}^l \mathbb F_q[x]/(x^{m_j}-1)\}$ is an Euclidean LCD code in $\mathbb F_q^l$.
\end{theorem}
\begin{proof}
Let $\xi$ be a primitive $m$-th  root of unity with $m=\prod_{j=1}^l m_j$ and $\hat{m}_j=\frac{m}{m_j}$. For any $j\in \{1,\cdots, l\}$,
let $T_j=\{i\in \mathbb Z_{m_j}: c_j(\xi^{\hat{m}_j i})=0\}$ and $S_j=\mathbb Z_{m_j}\setminus T_j$.

Let $\mathcal C$ be $\mu_{a}$-LCD. By Corollary \ref{cor:mu-a-iq}, for any $j\in \{1,\cdots, l\}$ and $i\in S_j$ with $i\neq 0$, one has
\begin{align*}
C_{\hat{m}_j i} \cap C_{-a\hat{m}_j i}^{\perp'}=\{0\}.
\end{align*}
Since $\mathbf{gcd}(m_i, m_j)=1$ for $1\le i< j\le l$, $C_{\hat{m}_j i} \cap C_{-a\hat{m}_j i}^{\perp'}=\{0\}$ for $i\in S_j\setminus \{0\}$  if and only if $-aS_j=S_j$.
From Theorem \ref{thm:mu-a-0-Vi}, $\mathcal C_j$ is $\mu_a$-LCD.

By Corollary \ref{cor:mu-a-iq} and using the fact that $\mathbb V_0=\mathbb F_q^l$, we deduce that $C_0$ is an Euclidean LCD code in $\mathbb F_q^l$.
Conversely, assume that $C_0$ is an Euclidean LCD code in $\mathbb F_q^l$ and $\mathcal C_j$ is $\mu_a$-LCD for any $j\in \{1, \cdots, l\}$.
It is observed that $\mathbb V_i=\{0\}$
for $i\in \mathbb Z_m^* \setminus (\cup_{j=1}^l S_j)$, where $\mathbb Z_m^*=\{x\in \mathbb Z_m: x\neq 0\}$. In particular, $C_{i}\cap C_{-ia}^{\perp'}=\{0\}$.
Let $\hat{m}_j i \in S_j \setminus \{0\}$. Since  $\mathcal C_j$ is $\mu_a$-LCD, from Theorem \ref{thm:mu-a-0-Vi} $-am_j i\in S_j$.
Note that
\begin{align*}
\mathbb V_{-a\hat{m}_ji}\cong \{0\}^{l-1}\times \mathbb F_q[\xi^{-a\hat{m}_ji}]  \text{ and }   C_{-a\hat{m}_j i}=\mathbb V_{-a\hat{m}_ji}.
\end{align*}
Then, $(C_{-a\hat{m}_j i})^{\perp'}=\{0\}$ and  $C_{\hat{m}_j i} \cap (C_{-a\hat{m}_j i})^{\perp'}=\{0\}$. Hence, $C_i \cap C_{-ai}^{\perp'}=\{0\}$ for any $i\in \mathbb Z_m$, which completes the proof.

\end{proof}

From Corollary \ref{cor:0-Vi-LCD} and Theorem \ref{thm:mi-mj-lcd}, we derive the following result which shows  that certain Euclidean LCD codes give rise to $\sigma$-LCD GCD codes and vice versa.

\begin{corollary}
Let $m_1, \cdots, m_l$ be positive integers with $\mathbf{gcd}(m_j,m_k)=1$ for $1\le j < k\le l$ and $q$ be a prime power  such that $\mathbf{gcd}(\prod_{j=1}^{l} m_j, q)=1$.
Then, $\mathcal C$ is a $\mu_{-1}$-LCD GQC codes  if and only if $C_0=\{(c_1(1), \cdots,
c_l(1)): (c_1(x), \cdots, c_l(x)) \in \prod_{j=1}^l \mathbb F_q[x]/(x^{m_j}-1)\}$ is an Euclidean LCD code in $\mathbb F_q^l$.
\end{corollary}

\subsection{ $\sigma$-LCD and $\sigma$ self-orthogonal  for $1$-generator GQC codes}

Let $\mathbf c(x)=(c_1(x), \cdots, c_{l}(x))$, where  $c_j(x)\in \mathbb F_q[x]/(x^{m_j}-1)$ for $j\in \{1,2, \cdots l\}$. Then
\begin{align*}
\mathbb F_q[x] \mathbf c(x)=\{(c(x)c_1(x), \cdots, c(x)c_{l}(x)): c(x)\in \mathbb F_q[x]\}
\end{align*}
is called a \emph{1-generator GQC code} with generator $\mathbf c(x)$. The $1$-generator GQC codes  are the most common GQC codes studied in the literature (see for instance \cite{Cy11, Seg04,PZ07}). The purpose of this subsection is to describe more explicitly $\sigma$-LCD and $\sigma$ self-orthogonal for $1$-generator GQC codes. Let $a$ be a positive integer and  $\mu_a$ be the transformation defined as above. We start with the following result which provides a necessary and sufficient condition for a $1$-generator GQC code to be $\mu_a$-LCD and $\mu_a$ self-orthogonal, respectively.
\begin{proposition}\label{prop:1-generator}
Let $m_1, \cdots,  m_{l}$ be positive integers with $\mathbf{gcd}(\prod_{j=1}^l m_j,q)=1$. Let $\mathcal C=\mathbb F_q[x] \mathbf c(x)$ be a  $1$-generator GQC code with $\mathbf c(x)\in \prod_{j=1}^l \mathbb F_q[x]/(x^{m_j}-1)$. Then $\mathcal C$ is  $\mu_a$-LCD if and only if  for any
$i\in \Gamma(q,m)$ with $(\delta_{1,i} c_1(\xi^i), \cdots,  \delta_{l,i} c_l(\xi^i))\neq 0$, the following condition holds:
\begin{align*}
\delta_{1,i} c_1(\xi^i) c_1(\xi^{-ai}) +\cdots + \delta_{l,i} c_l(\xi^i) c_l(\xi^{-ai})\neq 0.
\end{align*}
Further, $\mathcal C$ is  $\mu_a$ self-orthogonal if and only if for any
$i\in \Gamma(q,m)$, the following identity holds:
\begin{align*}
\delta_{1,i} c_1(\xi^i) c_1(\xi^{-ai}) +\cdots + \delta_{l,i} c_l(\xi^i) c_l(\xi^{-ai})= 0.
\end{align*}
\end{proposition}
\begin{proof}
For the 1-generator GQC code $\mathcal C=\mathbb F_q[x] \mathbf c(x)$, one has
\begin{align*}
C_i=\{\alpha (\delta_{1,i} c_1(\xi^i), \cdots,  \delta_{l,i} c_l(\xi^i)): \alpha \in \mathbb F_q[\xi^i] \},
\end{align*}
and
\begin{align*}
C_{-ai}=\{\alpha (\delta_{1,i} c_1(\xi^{-ai}), \cdots,  \delta_{l,i} c_l(\xi^{-ai})): \alpha \in \mathbb F_q[\xi^i] \}.
\end{align*}
Thus, $C_{i} \cap C_{-ai}^{\perp'} =\emptyset$ if and only if  for any $i\in \Gamma(q,m)$,
 $(\delta_{1,i} c_1(\xi^i), \cdots,  \delta_{l,i} c_l(\xi^i))=0$ or $ \sum_{j=1}^l \delta_{j,i} c_j(\xi^i) c_j(\xi^{-ai})\neq 0$.
Further, $C_{i} \subseteq C_{-ai}^{\perp'}$ if and only if $ \sum_{j=1}^l \delta_{j,i} c_j(\xi^i) c_j(\xi^{-ai})= 0$ for any $i\in \Gamma(q,m)$.
 The proof is then completed  from Corollary \ref{cor:mu-a}.
\end{proof}

The following result can be deduced from Proposition \ref{prop:1-generator} with $a=-1$ (we shall use the fact $\delta_{1,i} c_1(\xi^i)^2 +\cdots + \delta_{l,i} c_l(\xi^i) ^2=(\delta_{1,i} c_1(\xi^i) +\cdots +\delta_{l,i} c_l(\xi^i))^2$).

\begin{corollary}\label{cor:1-gene-LCD}
Let $m_1, \cdots,  m_{l}$ be $l$ positive integers with $\mathbf{gcd}(\prod_{j=1}^l m_j,q)=1$. Let $\mathcal C=\mathbb F_q[x] \mathbf c(x)$ be a  $1$-generator GQC code with $\mathbf c(x)\in \prod_{j=1}^l \mathbb F_q[x]/(x^{m_j}-1)$. Then
$\mathcal C$ is  $\mu_{-1}$-LCD  (resp. $\mu_{-1}$ self-orthogonal) if and only if $\delta_{1,i} c_1(\xi^i)^2 +\cdots +\delta_{l,i} c_l(\xi^i) ^2\neq 0$ for any
$i\in \Gamma(q,m)$ with $(\delta_{1,i} c_1(\xi^i), \cdots,  \delta_{l,i} c_l(\xi^i))\neq 0$
(resp. $\delta_{1,i} c_1(\xi^i)^2 +\cdots + \delta_{l,i} c_l(\xi^i) ^2= 0.$  for any
$i\in \Gamma(q,m)$).

Furthermore, when $q$ is a power of $2$, then $\mathcal C$ is  $\mu_{-1}$-LCD (resp. $\mu_{-1}$ self-orthogonal) if and only if
$\delta_{1,i} c_1(\xi^i) +\cdots +\delta_{l,i} c_l(\xi^i) \neq 0$ for any
$i\in \Gamma(q,m)$ with $(\delta_{1,i} c_1(\xi^i), \cdots,  \delta_{l,i} c_l(\xi^i))\neq 0$
(resp. $\delta_{1,i} c_1(\xi^i) +\cdots+ \delta_{l,i} c_l(\xi^i) = 0.$ for any $i\in \Gamma(q,m)$).
\end{corollary}
The following results provide characterizations of LCD $1$-generator QC codes and self-orthogonal $1$-generator GC codes.

\begin{corollary}\label{cor:1-generator-poly}
Let $m$ be a positive integer with $\mathbf{gcd}(m,q)=1$. Let $\mathcal C=\mathbb F_q[x] \mathbf c(x)$ be a $1$-generator QC code of index $l$ with $\mathbf c(x)=(c_1(x), \cdots, c_l(x))\in
 \left(\mathbb F_q[x]/(x^{m}-1)\right)^l$. Then $\mathcal C$ is  $\mu_{a}$-LCD (resp. $\mu_a$ self-orthogonal) if and only if $\mathbf{gcd}(\sum_{j=1}^l c_j(x) c_j(x^{-a \pmod{m}}), x^m-1)=\mathbf{gcd}(c_1(x), \cdots, c_l(x), x^m-1)$ (resp. $\sum_{j=1}^l c_j(x) c_j(x^{-a \pmod{m}}) \equiv 0  \pmod{ x^m-1}$).
 \end{corollary}
 \begin{proof}
 Firstly, recall the basic fact
  $$\mathbf{gcd}(f(x), x^m-1)=\prod_{i\in \mathbb Z_m, f(\xi^i)=0} (x-\xi^i),$$
  where $f(x)\in \mathbb F_q[x]$. Then, the proof follows from Corollary \ref{cor:1-gene-LCD}.
 \end{proof}

\begin{corollary}\label{cor:Si-Sj}
Let $m$ be a positive integer with $\mathbf{gcd}(m,q)=1$.  Let $\mathcal C=\mathbb F_q[x] \mathbf c(x)$ be a $1$-generator QC code of index $l$ with $\mathbf c(x)=(c_1(x), \cdots, c_l(x))\in
 \left(\mathbb F_q[x]/(x^{m}-1)\right)^l$ and $S_j=\{i\in \mathbb Z_m: c_j(\xi^i)\neq 0\}$ for $j\in \{1,2, \cdots, l\}$.
If $S_{i} \cap S_j= \emptyset$ for $1\le i< j\le l$, $\mathcal C=\mathbb F_q[x] (c_1(x), \cdots, c_l(x))$ is $\mu_{-1}$-LCD.
\end{corollary}
\begin{proof}
First, by choosing $a=-1$ in Corollary \ref{cor:1-generator-poly}, we observe that $\mathcal C$ is  $\mu_{-1}$-LCD (resp. $\mu_{-1}$ self-orthogonal) if and only if
 $\mathbf{gcd}(\sum_{j=1}^l c_j(x) ^2, x^m-1)=\mathbf{gcd}(c_1(x), \cdots, c_l(x), x^m-1)$ (resp. $\sum_{j=1}^l c_j(x) ^2 \equiv 0  \pmod{ x^m-1}$).
 Further, when $q$ is a power of $2$, then,  $\mathcal C$ is  $\mu_{-1}$-LCD (resp. $\mu_{-1}$ self-orthogonal) if and only if
 $\mathbf{gcd}(\sum_{j=1}^l c_j(x), x^m-1)=\mathbf{gcd}(c_1(x), \cdots, c_l(x), x^m-1)$ (resp. $\sum_{j=1}^l c_j(x)  \equiv 0  \pmod{ x^m-1}$).

Now, by $S_{i} \cap S_j= \emptyset$ for $1\le i< j\le l$, one has
\begin{align*}
\mathbf{gcd}(c_1(x), \cdots, c_l(x), x^m-1)= \prod_{i\in \mathbb Z_m \setminus \left( \cup_{j=1}^l S_j\right) } (x-\xi^i),
\end{align*}
and
\begin{align*}
\mathbf{gcd}(\sum_{j=1}^l c_j(x) ^2, x^m-1)= \prod_{i\in \mathbb Z_m \setminus \left( \cup_{j=1}^l S_j\right) } (x-\xi^i),
\end{align*}
which completes the proof.
\end{proof}
\begin{example}
Let $m$ be an odd prime with $m\equiv \pm 1 \pmod{8}$. Let
$c_1(x)=\delta+\sum_{i\in \mathbb Z_m^2\setminus \{0\} } x^i$
and
$c_2(x)=\delta+\sum_{i\in \mathbb Z_m\setminus Z_m^2 } x^i$,

where $\delta=0$ if $p\equiv 1 \pmod{8}$ and $\delta=1$ if $p\equiv -1 \pmod{8}$. Then,  $c_1(x)$ and $c_2(x)$ are the  generating idempotents of
 the even-like binary quadratic residue codes. Let $S_1$ and $S_2$ be defined as in Corollary \ref{cor:Si-Sj}. It is observed that
 $\{S_1, S_2\}=\{Z_m^2\setminus \{0\}, \mathbb Z_m\setminus Z_m^2 \}$  \cite{HP10}.
 From Corollary \ref{cor:Si-Sj}, $\mathcal C=\mathbb F_q[x] (c_1(x), c_2(x))$ is $\mu_{-1}$-LCD.
\end{example}

A $1$-generator QC code $\mathcal  C$ is called \emph{maximal} if any $1$-generator QC code $\mathcal C'$ containing $\mathcal C$ satisfies $\mathcal C=\mathcal C'$.
 It has been shown in \cite{SD90} that $\mathcal C=\mathbb F_q[x](c_1(x), \cdots, c_l(x))$ is a maximal $1$-generator QC code if and only if $\mathbf{gcd}(c_1(x), \cdots, c_l(x), x^m-1)=1$.
The following result characterizes all maximal $1$-generator QC LCD codes of index $l$.
\begin{corollary}\label{cor:maximal-1-generator-poly}
Let $m$ be a positive integer with $\mathbf{gcd}(m,q)=1$. Let $\mathcal C=\mathbb F_q[x] \mathbf c(x)$ be a maximal
$1$-generator QC code of index $l$ with $\mathbf c(x)=(c_1(x), \cdots, c_l(x))\in
 \left(\mathbb F_q[x]/(x^{m}-1)\right)^l$. Then $\mathcal C$ is  $\mu_{a}$-LCD if and only if
 $$\mathbf{gcd}(\sum_{j=1}^l c_j(x) c_j(x^{-a \pmod{m}}), x^m-1)=1.$$
 \end{corollary}
 \begin{proof}
The proof comes from  Corollary \ref{cor:1-generator-poly} and the fact that $\mathbf{gcd}(c_1(x), \cdots, c_l(x), x^m-1)=1$.
 \end{proof}

From Corollary \ref{cor:maximal-1-generator-poly}, we derive the following main result on
LCD maximal $1$-generator QC codes.
\begin{theorem}
Let $m$ be an odd positive integer and $q$ be a power of $2$. Then, a maximal $1$-generator QC code $\mathcal C$ of index $2$ over $\mathbb F_q[x]/(x^m+1)$ is $\mu_{-1}$-
LCD  if and only if there exists a unique $c(x)\in \mathbb F_q[x]/(x^m+1)$ such that $\mathcal C=\mathbb F_q[x](c(x), c(x)+1)$.
Further, the total number of $\mu_{-1}$-LCD maximal $1$-generator QC codes in $\left( \mathbb F_q[x]/(x^m+1) \right)^2$  is $q^m$.
\end{theorem}
\begin{proof}
If $\mathcal C=\mathbb F_q[x](c(x), c(x)+1)$, then from Corollary \ref{cor:maximal-1-generator-poly}, $\mathcal  C$ is $\mu_{-1}$-LCD maximal $1$-generator QC code.

Assume that $\mathcal C$ is a $\mu_{-1}$-LCD  maximal $1$-generator QC code in $(\mathbb F_q[x]/(x^m+1))^2$. Then, there exist
$c_1(x), c_2(x) \in (\mathbb F_q[x]/(x^m+1)$ such that $\mathcal C= \mathbb F_q[x](c_1(x), c_2(x))$ and $\mathbf{gcd}(c_1(x), c_2(x))=1$. Since $\mathcal  C$ is
$\mu_{-1}$-LCD, from Corollary \ref{cor:maximal-1-generator-poly}, $\mathbf{gcd}(c_3(x), x^m+1)=1$,  where $c_3(x)=c_1(x)+c_2(x)$. Thus, there is a unique $c_3'(x) \in \mathbb F_q[x]/(x^m+1)$
such that $c_3(x)c_3'(x) \equiv 1 \pmod{x^m+1}$, that is, $c_3(x)$ and $c_3'(x)$ are invertible in $\mathbb F_q[x]/(x^m+1)$. Then,
\begin{align*}
\mathcal C&= \{a(x)(c_1(x), c_2(x)): a(x) \in \mathbb F_q[x]/(x^m+1) \} \\
&=\{a(x)(c_1(x), c_1(x)+c_3(x)): a(x) \in \mathbb F_q[x]/(x^m+1) \}\\
&=\{a(x)c_3'(x)(c_1(x), c_1(x)+c_3(x)): a(x) \in \mathbb F_q[x]/(x^m+1) \}\\
&=\{a(x)(c(x), c(x)+1): a(x) \in \mathbb F_q[x]/(x^m+1) \},
\end{align*}
where $c(x)=c_1(x)c_3'(x)$. Hence, $\mathcal C=\mathbb F_q[x](c(x), c(x)+1)$.

If there is some $b(x)\in \mathbb F_q[x]/(x^m+1)$ such that
$\mathcal C=\mathbb F_q[x](c(x), c(x)+1)=\mathbb F_q[x](b(x), b(x)+1)$, one gets
\begin{align*}
c(x)(b(x)+1)\equiv b(x)(c(x)+1) \pmod{x^m+1}.
\end{align*}
Then, $b(x)\equiv c(x) \pmod{x^m+1}$, which completes the proof.
\end{proof}

\subsection{Asymptotically good $\mu_{-1}$-LCD GQC codes}
In this subsection, we will  construct a class of $\mu_{-1}$-LDC GQC codes and prove the existence of  asymptotically good $\mu_{-1}$-LDC GQC codes.

Let $q$ be a prime power and assume that $m_1, m_2, \cdots , m_l$ are pairwise distinct
positive integers relatively prime to $q$. Let $m$ be  the least common multiple of $m_1, m_2, \cdots, m_{l}$  and
$\xi$ be a primitive $m$-th root of unity in the  algebraic closure of $\mathbb F_q$. Then,
$\xi^{\hat{m}_j}$ is a primitive $m_j$-th root of unity, where $\hat{m}_j= \frac{m}{m_j}$ and $j\in \{1,2, \cdots, l\}$. Moreover,
\begin{align*}
\mathbb F_q[\xi^{\hat{m}_j}] \cong \mathbb F_q[x]/(\mathbb M_{\xi^{\hat{m}_j}}(x)) \cong \mathbb (\mathbb H_j(x) \mathbb F_q[x])/(x^{m_j}-1),
\end{align*}
where $\mathbb M_{\xi^{\hat{m}_j}}(x)$ is defined as in Equation (\ref{eq:M-xi-i}) and $\mathbb H_j(x) =\frac{x^{m_j}-1}{\mathbb M_{\xi^{\hat{m}_j}}(x)}$.
Hence, any linear code $C_{\hat{m}_j}$ in  $(\mathbb F_q[\xi^{\hat{m}_j}])^{r_j}$ can be identified as an $\mathbb F_q[x]$ submodule of
$(\mathbb (\mathbb H_j(x) \mathbb F_q[x])/(x^{m_j}-1))^{r_j}$, where $r_j$ is a positive integer.

For any $j\in \{1,2, \cdots, l\}$, let $C_{\hat{m}_j}$ be an Euclidean LCD code  over $\mathbb (\mathbb H_j(x) \mathbb F_q[x])/(x^{m_j}-1)$  with parameters $[r_j, k_j, d_j]$.
Now, we define the GQC code by
\begin{align}\label{eq:C-C1-Cl}
\mathcal C&= C_{\hat{m}_1}\times \cdots  \times C_{\hat{m}_l} \nonumber \\
&=\{(\mathbf c_1(x), \cdots, \mathbf c_l(x)):   \mathbf c_j(x) \in  C_{\hat{m}_j} \text{ for } j\in \{1,\cdots, l\}  \}.
\end{align}

Note that $\mathbb H_j(\xi^i) \neq 0$ if and only if $i\in \{\hat{m}_j q^k: j\in \{1, \cdots, l\}, k\in \{1, \cdots, m\}\}$.
From the definition of $\mathcal C$,  one observes that
\begin{align*}
C_i=\begin{cases}
(C_{\hat{m}_j})^{q^k}\times \prod_{k\in \{1, \cdots, l\} \setminus \{j\}} \underbrace{\{0\}, \cdots, \{0\}}_{r_k}, &  i=\hat{m}_j q^k
\cr 0, & \text{otherwise}
\end{cases}
\end{align*}
where $C_i$ ($i\in \mathbb Z_m$) are the constituents of $\mathcal C$ and $(C_{\hat{m}_j})^{q^k}=\{(c_1^{q^k}, \cdots, c_l^{q^k}): (c_1, \cdots, c_l)\in C_{\hat{m}_j}\}$.
Since $C_{\hat{m}_j}$ ($j\in \{1, \cdots, l\}$) are Euclidean LCD in $(\mathbb F_q[\xi^{\hat{m}_j}])^{r_j}$, $C_i \cap C_{i}^{\perp'}=\{0\}$ for $i\in \mathbb Z_m$.
From Corollary \ref{cor:mu-a-iq}, $\mathcal C$ is a $\mu_{-1}$-LCD GQC code in $\prod_{j=1}^{l} (\mathbb F_q[x]/(x^{m_j}-1))^{r_j}$. By Equation (\ref{eq:C-C1-Cl}), one observes that
$\mathcal C$ is
an $\mathbb F_q$-linear code with parameters
$\left [\sum_{j=1}^l m_j r_j, \sum_{j=1}^l k_j deg(\mathbb M_{\xi^{\hat{m}_j}}(x)),  d_{\mathcal C}\ge \min {\{d_1d_{\mathbb H_1}, \cdots, d_ld_{\mathbb H_l}\}}\right ]$, where $d_{\mathcal C}$
is the minimum distance of $\mathcal C$ and $d_{\mathbb H_j}$ is the minimum distance of the cyclic code $(\mathbb H_j(x) \mathbb F_q[x])/(x^{m_j}-1)$ for
$j\in \{1, \cdots, l\}$.
Thus,
\begin{align*}
R_{\mathcal C}= \frac{\sum_{j=1}^l k_j deg(\mathbb M_{\xi^{\hat{m}_j}}(x))}{\sum_{j=1}^l m_j r_j},
\end{align*}
and
\begin{align*}
\delta_{\mathcal C}= \frac{d_{\mathcal C}}{\sum_{j=1}^l m_j r_j}\ge \frac{\min {\{d_1d_{\mathbb H_1}, \cdots, d_ld_{\mathbb H_l}\}}}{\sum_{j=1}^l m_j r_j}.
\end{align*}
If $r_1=\cdots=r_l=r$, then
\begin{align}\label{eq:R-C}
R_{\mathcal C}= \sum_{j=1}^l \frac{ deg(\mathbb M_{\xi^{\hat{m}_j}}(x))}{\sum_{j=1}^l m_j} \frac{k_j}{r},
\end{align}
and
\begin{align}\label{eq:delta-C}
\delta_{\mathcal C}\ge     \frac{ \min{\{d_{\mathbb H_1}, \cdots,
d_{\mathbb H_l}\}}}{\sum_{j=1}^l m_j}   \min {\{\frac{d_1}{r}, \cdots, \frac{d_l}{r}\}}.
\end{align}
The above development yields the following main result of this section. We shall use Equations (\ref{eq:R-C}) and (\ref{eq:delta-C}) to show the existence of
 asymptotically good long $\mu_{-1}$-LCD GQC codes.
\begin{theorem}
Let $m_1, \cdots,  m_{l}$ be pairwise distinct positive integers with $\mathbf{gcd}(\prod_{j=1}^l m_j,q)=1$.
Then there exists an asymptotically good sequence of $\mu_{-1}$-LCD GQC codes $\mathcal \{\mathcal C^{(i)}\}_{i=1}^{\infty}$,
where $\mathcal C^{(i)}$ is a  GQC code in $\prod_{j=1}^l (\mathbb F_q[x]/(x^{m_j}-1))^{r^{(i)}}$ and $r^{(i)}$ is a positive integer.
\end{theorem}
\begin{proof}
Note that the asymptotically good sequence of Euclidean LCD codes over $\mathbb F_q[\xi^{\hat{m}_j}]$ exists by \cite{Sen04}. Thus, there exist positive constants $R_0$, $\delta_0$ and
Euclidean LCD codes $C_{\hat{m}_j}^{(i)}$ over $\mathbb F_q[\xi^{\hat{m}_j}]$  with parameters $[r^{(i)},k^{(i)}_j, d^{(i)}_j]$ such that
$\frac{k^{(i)}_j}{r^{(i)}}\ge R_0$ and $\frac{d^{(i)}_j}{r^{(i)}}\ge \delta_0$, where $j\in \{1,2, \cdots, l\}$ and
$i\in \{1,2, \cdots, +\infty\}$. Let $\mathcal C^{(i)}=\prod_{j=1}^l C_{\hat{m}_j}^{(i)}$ be the $\mu_{-1}$-LCD GQC code defined as in Equation (\ref{eq:C-C1-Cl}).
From Equations (\ref{eq:R-C}) and (\ref{eq:delta-C}), one obtains
\begin{align*}
R_{\mathcal C^{(i)}}\ge \frac{ \sum_{j=1}^l deg(\mathbb M_{\xi^{\hat{m}_j}}(x))}{\sum_{j=1}^l m_j} R_0
\end{align*}
and
\begin{align*}
\delta_{\mathcal C^{(i)}}\ge \frac{ \min{\{d_{\mathbb H_1}, \cdots,
d_{\mathbb H_l}\}}}{\sum_{j=1}^l m_j}   \delta_0.
\end{align*}
The desired result follows.

\end{proof}

\section{$\mu_{-1}$-LCD Abelian  codes}\label{sec:Abelian-LCD}
In \cite{CW16}, Cruz and Willems have investigated and characterized ideals in a group algebra which have Euclidean complementary duals. In this section, we shall consider
$\mu_{-1}$-LCD Abelian codes and show that all Abelian codes in  semi-simple communicative group algebra are $\mu_{-1}$-LCD.

Let $G$ be a finite Abelian group with identity $1$ and $\mathbb F_q[G]$ be the group algebra of $G$ over $\mathbb F_q$.
The elements in $\mathbb F_q[G]$ will be written as  $\sum_{g\in G} c_g g$, where $c_g\in \mathbb F_q$. The addition and the multiplication in $\mathbb F_q[G]$
are given as follows.
\begin{align*}
\left (\sum_{g\in G} a_g g \right )+\left (\sum_{g\in G} b_g g \right )=\sum_{g\in G} (a_g+b_g) g,
\end{align*}
and
\begin{align*}
\left (\sum_{g\in G} a_g g \right )\cdot \left (\sum_{g\in G} b_g g \right )=\sum_{g\in G} \left (\sum_{h\in G} a_h b_{gh^{-1}} \right) g,
\end{align*}
where $\sum_{g\in G} a_g g, \sum_{g\in G} b_g g  \in \mathbb F_q[G]$.
We define $\mu_{-1}$ on $\mathbb F_q[G]$ to be the map that sends $\sum_{g\in G} a_g g$ to $\sum_{g\in G} a_{g^{-1}} g$ for $\sum_{g\in G} a_g g \in \mathbb F_q[G]$.
Note that $\mu_{-1}: \mathbb F_q[G] \longrightarrow \mathbb F_q[G]$ defines an auto-isomorphism of $\mathbb F_q[G]$.

A linear code $\mathcal C$ is called an \emph{Abelian code} (for $G$ over the field $\mathbb F_q$) if $\mathcal  C$ is an ideal in the
group algebra $\mathbb F_q[G]$.

For $\mathbf b =\sum_{g\in G} b_g g , \mathbf c= \sum_{g\in G} c_g g \in \mathbb F_q[G]$, the Euclidean inner product and $\mu_{-1}$ inner product of $\mathbf b$ and $\mathbf c$
are respectively defined as follows.
\begin{align*}
<\mathbf b, \mathbf c>&=\sum_{g\in G} b_g c_g,\\
<\mathbf b, \mathbf c>_{\mu_{-1}}&=\sum_{g\in G} b_g c_{g^{-1}}.
\end{align*}
Then, $<\cdot, \cdot>$ is $G$-invariant, that is
\begin{align*}
<g \mathbf b,  g \mathbf c>=<\mathbf b,  \mathbf c>,
\end{align*}
where $g\in G$, and $\mathbf b,  \mathbf c\in \mathbb F_q[G]$. In particular, $<g \mathbf b,   \mathbf c>=<\mathbf b, g^{-1}  \mathbf c>= <\mathbf b, \mu_{-1}(g)  \mathbf c>$.
Thus, for any $\mathbf a, \mathbf b,  \mathbf c\in \mathbb F_q[G]$,
\begin{align}\label{eq:G-invariant}
<\mathbf a \mathbf b,   \mathbf c>= <\mathbf b, \mu_{-1}(\mathbf a)  \mathbf c>.
\end{align}

We have the following main result of this section.
\begin{theorem}\label{thm:idempotent}
Let $\mathcal C$ be an Abelian code in $\mathbb F_q[G]$. Then, $\mathcal  C$ is $\mu_{-1}$-LCD if and only if $\mathcal C$ is generated by an idempotent.
\end{theorem}
\begin{proof}
First suppose that $\mathcal{C}$ is $\mu_{-1}$-LCD, that is, $ \mathcal C \cap (\mu_{-1}(\mathcal C))^{\perp}= \{0\}$  and
 $\mathcal C + (\mu_{-1}(\mathcal C))^{\perp}=\mathbb F_q[G]$. Thus, there exist $\mathbf e\in \mathcal C$ and $\mathbf f\in \mathcal C$
 such that $\mathbf e+\mathbf f=1$. From $\mathbf e+\mathbf f=1$, $\mathbb F_q[G] \mathbf e+ \mathbb F_q[G] \mathbf f=\mathbb F_q[G]$. One has $\mathcal C=\mathbb F_q[G] \mathbf e$.
Since $\mathbf e-\mathbf e^2=\mathbf e\mathbf f \in  \mathcal C \cap (\mu_{-1}(\mathcal C))^{\perp}$, one obtains $\mathbf e^2=\mathbf e$.

Conversely suppose that $\mathcal C=\mathbb F_q[G] \mathbf e$, where $\mathbf e^2=\mathbf e$. Let
$\mathbf c \mathbf e\in \mathcal C \cap (\mu_{-1}(\mathcal C))^{\perp}$. Then, for any $\mathbf b\in \mathbb F_q[G]$,
\begin{align*}
0=&<\mathbf c \mathbf e,  \mu_{-1}(\mathbf e\mathbf b)>\\
=& <\mathbf c \mathbf e^2, \mu_{-1}(\mathbf b)>\\
=&<\mathbf c \mathbf e, \mu_{-1}(\mathbf b)>,
\end{align*}
where the second identity follows from Equation (\ref{eq:G-invariant}).
Thus, $\mathbf c \mathbf e \in (\mu_{-1}(\mathbb F_q[G]))^{\perp}=\mathbb F_q[G]^{\perp}$. Hence, $\mathbf c \mathbf e=0$, which shows
$\mathcal C \cap (\mu_{-1}(\mathcal C))^{\perp}=\mathcal C \cap (\mathcal C)^{\perp_{\mu_{-1}}}=\{0\}$, that is,
$\mathcal  C $ is $\mu_{-1}$-LCD.

\end{proof}
The following consequence of Theorem \ref{thm:idempotent} shows that Abelian codes in the group algebra $\mathbb F_q[G]$ of $G$ (with $\mathbf{gcd}(\# G, q)=1$) are particular
$\sigma$-LCD codes.

\begin{corollary}
Let $q$ be a prime power and  $G$ be a  finite Abelian group with $\mathbf{gcd}(\# G, q)=1$. Then, all Abelian codes in $\mathbb F_q[G]$ are $\mu_{-1}$-LCD.
\end{corollary}
\begin{proof}
Since $\mathbf{gcd}(\# G, q)=1$, the group algebra $\mathbb F_q[G]$ is semi-simple and  all ideals of $\mathbb F_q[G]$ are of the form $\mathcal C=\mathbb F_q[G] \mathbf e$,
where $\mathbf e \in  \mathbb F_q[G]$ is an idempotent element \cite{Ber67}. It completes the proof from Theorem \ref{thm:idempotent}.
\end{proof}

\begin{example}
Let $\mathcal G_{23}$ be the $[23, 12, 7]$  binary Golay code. Since $\mathcal G_{23}$ is an Abelian code in $\mathbb F_2[\mathbb Z_{23}]$,   $\mathcal G_{23}$ is a
$\mu_{-1}$-LCD codes with parameters $[23, 12,7]$. However there is no Euclidean LCD code with parameters $[23,12,7]$, since
any $[23,12,7]$ linear codes are permutational  equivalent to $\mathcal G_{23}$, which are not Euclidean LCD \cite{MS77}.
\end{example}

\section{Concluding Remarks}
In this paper, we first introduce
 the concept of   $\sigma$-LCD codes which generalizes the well-known Euclidean LCD codes, Hermitian LCD codes, and Galois LCD codes.
 Likewise Euclidean LCD codes, $\sigma$-LCD codes can also be used to construct LCP of codes. We show that, all $q$-ary ($q>2$) linear codes are $\sigma$-LCD and any binary linear code $\mathcal C$,  $\{0\}\times \mathcal C$ is  $\sigma$-LCD. Furthermore, we develop the theory of $\sigma$-LCD GQC codes by providing characterizations and constructions of asymptotically good $\sigma$-LCD GQC codes. Finally, we consider $\sigma$-LCD Abelian codes and prove  that all Abelian codes in a semi-simple group algebra are $\sigma$-LCD.
 Our results show that $\sigma$-LCD codes allow the construction of LCP of codes more easily and with more flexibility. Moreover, many results on the classical LCD codes and LCD GQC codes can be derived from our results on $\sigma$-LCD codes by fixing appropriate mappings $\sigma$.

\ifCLASSOPTIONcaptionsoff
  \newpage
\fi

\end{document}